\documentclass[11pt,a4paper]{article}
\usepackage{amsmath,amssymb,amsthm}
\usepackage[margin=1in]{geometry}
\usepackage{xcolor}
\usepackage[colorlinks,citecolor={blue!50!black},linkcolor={red!50!black},urlcolor={blue!70!black}]{hyperref}

\newcommand{\Def}{\operatorname{Def}}
\newcommand{\Ric}{\operatorname{Ric}}
\newcommand{\Div}{\operatorname{div}}
\newcommand{\Divg}{\operatorname{div}_g}
\newcommand{\Curl}{\operatorname{curl}}
\newcommand{\curlg}{\operatorname{curl}_g}
\newcommand{\gradg}{\nabla_g}
\newcommand{\tr}{\operatorname{tr}}
\newcommand{\R}{\mathbb{R}}
\newcommand{\eps}{\varepsilon}
\newcommand{\II}{\mathrm{I\!I}}
\newcommand{\LapB}{\Delta_{B}}
\newcommand{\LapH}{\Delta_{H}}
\newcommand{\LapDef}{\Delta_{\Def}}
\newcommand{\Frad}{\mathcal{F}_{\mathrm{rad}}}
\newcommand{\Hh}{\mathcal H}

\newcommand{\Psig}{P_{\sigma}}
\newcommand{\norm}[1]{\lVert #1\rVert}

\theoremstyle{plain}
\newtheorem{theorem}{Theorem}[section]
\newtheorem{proposition}[theorem]{Proposition}
\newtheorem{lemma}[theorem]{Lemma}
\newtheorem{corollary}[theorem]{Corollary}
\theoremstyle{remark}

\title{Boundary conditions select the viscous operator on Riemannian
hypersurfaces: formal analysis and rigorous thin-shell limits}
\author{Zhi-Wei Wang\thanks{College of Physics, Jilin University, Changchun
130012, China; \texttt{zhiweiwang.phy@gmail.com}.}
\and
Samuel L.\ Braunstein\thanks{Department of Computer Science, University of York,
York YO10 5GH, United Kingdom; \texttt{sam.braunstein@york.ac.uk}.}}
\date{\today}

\begin{document}
\maketitle

\begin{abstract}
A viscous fluid confined to a thin layer around a curved surface is governed, as
the layer thickness vanishes, by an effective viscous operator on the surface. We
show that the wall conditions select this operator. The vorticity-free (free) and
stress-free (Navier) conditions, which coincide on flat walls but differ on
curved ones by the shape operator, yield respectively the Hodge Laplacian and the
deformation Laplacian, and these differ universally, on any hypersurface, by
twice the Ricci curvature; a one-parameter family of wall conditions joins them,
with an effective operator that couples to the extrinsic geometry only in
between. We prove this in two forms: formally, by matched asymptotics, on an
arbitrary hypersurface, and rigorously, as Mosco convergence of the viscous
energy forms, hence with resolvent, semigroup and spectral convergence, on
surfaces of revolution. The stress-free limit on general surfaces is due to Miura
and is recovered here; the rigorous vorticity-free limit beyond the sphere, via a
uniform Gaffney inequality, together with the interpolating family and the spectral
packaging, is new, and the analysis makes precise a conflation of the two
conditions in the classical spherical treatment. The extension-dependence of the
operator found on the ellipsoid is explained as a dependence on the wall
condition.
\end{abstract}

\section{Introduction}

A viscous incompressible fluid confined to a thin layer around a curved surface
is governed, as the layer thickness tends to zero, by an effective viscous
operator on the surface. This thin-film reduction underlies the modelling of
geophysical flows in the atmosphere and ocean, of hydrodynamics on biological
membranes, and of flows along fluid interfaces. The identity of the effective
operator has, however, proved subtle. Unlike the scalar Laplace-Beltrami
operator, the viscous operator acting on a tangential vector field on a curved
surface is not fixed by the intrinsic geometry alone, and different derivations
have produced different operators; which operator is correct, and what
determines it, is the subject of a recent literature surveyed by
Czubak~\cite{Czubak2024}.

The ambiguity has a clean origin in the boundary conditions. A fluid in a
three-dimensional layer is bounded by two walls, at which the physically natural
tangential conditions are of two kinds. The first is the vorticity-free, or
free, condition,
\begin{equation}\label{eq:free}
u\cdot n = 0,\qquad (\operatorname{curl}u)\times n = 0,
\end{equation}
and the second is the stress-free, or Navier, condition,
\begin{equation}\label{eq:navier}
u\cdot n = 0,\qquad \big(\Def(u)\,n\big)_{\mathrm{tan}} = 0,
\end{equation}
where $\Def(u)=\tfrac12(\nabla u+\nabla u^{\mathsf T})$ is the rate-of-deformation
tensor and $n$ is the wall normal. On a flat wall the conditions \eqref{eq:free}
and \eqref{eq:navier} coincide. On a curved wall they differ by the shape
operator (Weingarten map) of the wall, the difference being of the order of the
wall's curvature. This distinction is invisible in flat geometry, and it is the
crux of the present paper.

The thin-film limit has been approached from both sides of this distinction,
but the distinction itself has not been isolated. In flat layers
$\omega\times(0,\eps)$, Temam and Ziane~\cite{TZ96} treated a range of boundary
conditions, including the free condition \eqref{eq:free}, on whose flat walls
free and stress-free agree. For the sphere, Temam and Ziane~\cite{TZ97} derived
the surface Navier-Stokes equations as the thin-film limit of a spherical shell,
in justification of the primitive equations of geophysics; they imposed the free
condition \eqref{eq:free} and identified it with the stress-free condition
\eqref{eq:navier}, and their limit operator, recorded as the Laplace-Beltrami
operator, is in fact the Hodge-de Rham Laplacian. On general surfaces,
Miura~\cite{MiuraIII} gave the first rigorous derivation of the surface
Navier-Stokes equations, under the stress-free condition \eqref{eq:navier},
obtaining in the clean case the viscous operator
$2\Div_g\Def_g=\LapB+\Ric$, equivalently $\LapH+2\Ric$ on divergence-free
fields; he observed in a remark~\cite[Rem.~2.10]{MiuraIII} that this differs from
the spherical Hodge limit of~\cite{TZ97} by twice the Weingarten map, and hence,
on the sphere, by $2\Ric$. Most recently, Chan, Czubak and Fuster
Aguilera~\cite{CCF25}, extending the ellipsoidal restriction analysis of Chan,
Czubak and Yoneda~\cite{CCY23}, carried out a thin-shell asymptotic derivation of
the viscous operator on the ellipsoid and found that the operator obtained
depends on the method, that is, on the way the surface flow is extended into the
shell. The coercivity underlying these limits, a uniform-in-thickness Korn
inequality on the shell, was established for general closed hypersurfaces by
Lewicka and Müller~\cite{LM11}, and in the Navier setting by Miura~\cite{MiuraI},
in both cases with the inequality degenerating on the Killing fields of the
surface.

The common thread is that the effective operator depends on a choice, of boundary
condition or of extension, made where the two walls of the shell live, but this
dependence has been seen only case by case: the free condition only on the
sphere and there conflated with the stress-free one, the stress-free condition on
general surfaces, the ellipsoidal method-dependence as a computation on a single
surface. What has been missing is the principle behind these observations and its
status as a theorem. The present paper supplies it. We show that on an arbitrary
hypersurface the wall condition selects the effective viscous operator; that the
vorticity-free and stress-free conditions produce, respectively, the Hodge
Laplacian $\LapH$ and the deformation Laplacian $\LapDef=\LapB+\Ric$; and that
the two differ universally by $2K$, twice the Ricci curvature, the two endpoints
being joined by a one-parameter family of wall conditions whose effective
operators interpolate and couple to the extrinsic geometry only in between. We
establish this in two forms: rigorously, as convergence of the viscous energy
forms in the sense of Mosco, and hence with convergence of resolvents,
semigroups and spectra, on surfaces of revolution; and formally, by matched
asymptotic expansions, on general hypersurfaces. The extension-dependence
of~\cite{CCY23,CCF25} is thereby explained, as a boundary-condition dependence
whose magnitude is the same shape-operator term that separates the two
endpoints.

Our results are complementary to Miura's. The stress-free limit on a general
surface is his, and we recover it as one endpoint of the family; what is new is
the rigorous vorticity-free (Hodge) limit beyond the sphere, obtained through a
uniform
Gaffney inequality for the Hodge energy that we establish; the selection
principle stated as a theorem with an interpolating family, rather than a
comparison of two fixed cases; and the formulation as Mosco convergence, which
delivers the spectral consequences that solution-level convergence does not. The
uniform Korn inequality that the deformation energy requires we take
from~\cite{LM11,MiuraI}, proving it directly only in the revolution case treated
here. Where our analysis meets the classical spherical treatment it makes precise
what was there implicit: on a curved surface the free and stress-free conditions
are genuinely distinct, and the Laplace-Beltrami operator of the spherical limit
is specifically the Hodge Laplacian.

Our principal results are the operator dichotomy $L^{\mathrm{slip}}=L^{\mathrm{H}}
+2K$, established formally on a general hypersurface (Theorem~\ref{thm:formal}) and
rigorously on surfaces of revolution (Theorem~\ref{thm:dichotomy}); the Mosco
convergence of the viscous forms on surfaces of revolution, for both boundary
conditions, with its resolvent, semigroup and spectral corollaries
(Theorem~\ref{thm:mosco}); and the general formal decomposition and the
interpolating family on arbitrary hypersurfaces (Section~\ref{sec:formal}), with
the deformation Laplacian arising from the deformation-tensor structure of the
viscous stress in the sense of Ebin and Marsden~\cite{EM70}. The rigorous results
hold for torus-type surfaces of revolution, a curved, non-umbilic and
non-constant-curvature family that includes the spheroid, and the formal results
require no symmetry.

The paper is organised as follows. Section~\ref{sec:setup} fixes the geometry and
defines the two energy forms. Section~\ref{sec:formal} gives the formal
decomposition on a general hypersurface, the two-wall determination of the
effective operator, the interpolating family and the dichotomy, and the
solenoidal recovery expansion. Section~\ref{sec:rigorous} makes these rigorous on
surfaces of revolution: the uniform Korn and Gaffney inequalities, the exact
identification of the reduced forms, and Mosco convergence with its spectral
consequences. Section~\ref{sec:discussion} discusses the universality mechanism,
the relation to previous work, the interpolating family, and the outlook.

\section{Geometric preliminaries and the two energy forms}
\label{sec:setup}

\subsection{The surface and the shell}

Let $M=M^{n}$ be a smooth, closed, orientable hypersurface embedded in
$\R^{n+1}$, with unit normal field $N$ and induced metric $g$. The shape
operator $S\colon TM\to TM$, defined by $S(X)=-\bar\nabla_X N$ with $\bar\nabla$
the flat ambient connection, is $g$-self-adjoint; its eigenvalues are the
principal curvatures, its normalised trace $H=\tfrac1n\tr S$ is the mean
curvature, and $\II(X,Y)=g(SX,Y)$ is the second fundamental form. The Gauss
equation for a hypersurface in flat space expresses the intrinsic Ricci tensor
through the shape operator,
\begin{equation}\label{eq:gauss}
\Ric_{ij}=nH\,S_{ij}-(S^{2})_{ij}.
\end{equation}
In dimension $n=2$, on which the rigorous results below are set, the Ricci
tensor is $\Ric=Kg$ with $K$ the Gaussian curvature, so the Ricci term acts on a
tangent vector as multiplication by the scalar $K$; we use $\Ric$ in the general
statements and $K$ in the surface-specific ones.

The shell of thickness $\eps$ about $M$ is
\[
\Sigma_\eps=\bigl\{\,p+rN(p):p\in M,\ |r|<\tfrac{\eps}{2}\,\bigr\},
\]
described by Fermi (normal) coordinates $(x,r)$, $x\in M$, $|r|<\eps/2$. The
parallel hypersurface $M_r=\{r=\text{const}\}$ has, in the flat ambient space,
shape operator and induced metric
\begin{equation}\label{eq:riccati}
S(r)=S\,(\mathrm{Id}-rS)^{-1},\qquad
g_r(X,Y)=g\bigl((\mathrm{Id}-rS)X,\,(\mathrm{Id}-rS)Y\bigr),
\end{equation}
the first being the Riccati equation $\partial_r S=S^{2}$ integrated from $M$.
The ambient metric on the shell is therefore $\bar g=g_r\oplus dr^{2}$, and the
volume element is $\det(\mathrm{Id}-rS)\,dV_g\,dr$.

We rescale the normal variable onto a fixed cylinder. Setting $r=\eps z$ with
$z\in(-\tfrac12,\tfrac12)$, the shell becomes $\Sigma=M\times(-\tfrac12,\tfrac12)$,
carrying the metric and rescaled volume
\begin{equation}\label{eq:rescaled}
G_\eps=g_{\eps z}\oplus\eps^{2}\,dz^{2},\qquad
d\mu_\eps=\det(\mathrm{Id}-\eps z\,S)\,dV_g\,dz,
\end{equation}
where $d\mu_\eps$ is the volume per unit rescaled thickness; normal
differentiation carries a factor of the thickness, $\partial_r=\eps^{-1}
\partial_z$. A vector field on $\Sigma$ is written $U=U^{t}+U^{n}N$ in its
tangential and normal parts.

\subsection{The two energy forms}

The viscous dynamics on the shell is generated, for each wall condition, by a
quadratic energy form. With the viscosity normalised to $1$ and the energy taken
per unit rescaled thickness, the two forms are the deformation (stress-based)
energy and the Hodge energy,
\begin{align}
Q^{\mathrm{slip}}_\eps(U)&=\int_\Sigma 2\,\bigl|\Def_{G_\eps}U\bigr|^{2}\,d\mu_\eps,
\label{eq:Qslip}\\[2pt]
Q^{\mathrm H}_\eps(U)&=\int_\Sigma\Bigl(\bigl|\Curl U\bigr|^{2}
+\bigl(\Div U\bigr)^{2}\Bigr)\,d\mu_\eps,
\label{eq:Qhodge}
\end{align}
where $\Def_{G_\eps}U=\tfrac12\bigl(\bar\nabla U+\bar\nabla U^{\mathsf T}\bigr)$
is the rate-of-deformation tensor and $\Curl$, $\Div$ are taken in the shell
metric. Physically, $Q^{\mathrm{slip}}_\eps$ is twice the squared rate of
deformation, the Newtonian viscous dissipation; $Q^{\mathrm H}_\eps$ is the
energy whose Euler-Lagrange operator is the Hodge Laplacian.

The two forms carry the two wall conditions of the introduction through their
domains. Impermeability, $U^{n}=0$ on the walls $z=\pm\tfrac12$, is imposed in
both cases. The deformation form is defined on
\[
V_\eps=\bigl\{\,U\in H^{1}(\Sigma)^{n+1}:\ U^{n}=0\ \text{on}\ z=\pm\tfrac12\,\bigr\},
\]
its natural (variational) boundary condition being the stress-free (Navier)
condition $(\Def U\,n)_{\mathrm{tan}}=0$; the Hodge form is defined on
\[
V^{\mathrm H}_\eps=\bigl\{\,U\in H(\Curl)\cap H(\Div):\ U^{n}=0\ \text{on}\
z=\pm\tfrac12\,\bigr\},
\]
its natural boundary condition being the vorticity-free condition
$(\Curl U)\times n=0$. On the smooth two-wall boundary $V^{\mathrm H}_\eps$
coincides, with equivalent norms, with the $U^{n}$-constrained $H^{1}$ space
\cite{GiraultRaviart}. The viscous (Stokes-type) operator associated with each
form is obtained by restricting the form to its solenoidal subspace
$\{\Div_{G_\eps}U=0\}$.

\subsection{The candidate limit operators}

As $\eps\to0$ the two shell forms are expected to converge to surface energies
on tangential, divergence-free fields $v$ on $M$, namely the deformation and
Hodge energies
\begin{equation}\label{eq:limitforms}
Q^{\mathrm{slip}}_0(v)=2\int_M\bigl|\Def_g v\bigr|^{2}\,dV_g,
\qquad
Q^{\mathrm H}_0(v)=\int_M\bigl(|dv^{\flat}|^{2}+|\delta v^{\flat}|^{2}\bigr)\,dV_g,
\end{equation}
whose Euler-Lagrange operators on solenoidal fields are the deformation
Laplacian and the Hodge Laplacian,
\begin{equation}\label{eq:limitops}
L^{\mathrm{slip}}=\Psig\bigl(\LapH+2\Ric\bigr),\qquad
L^{\mathrm H}=\Psig\,\LapH,\qquad
\LapH=-(d\delta+\delta d),
\end{equation}
with $\Psig$ the Leray projection onto solenoidal fields. The two limit
operators thus differ by the Ricci term $2\Ric$, which on a surface is $2K$; this
is the operator dichotomy. That the two forms do converge to these limits, so
that the wall condition selects the operator through the sign with which the
shape operator enters, is the content of the paper: it is established formally on
a general hypersurface in Section~\ref{sec:formal}, and rigorously, as Mosco
convergence of \eqref{eq:Qslip}--\eqref{eq:Qhodge} to \eqref{eq:limitforms}, on
surfaces of revolution in Section~\ref{sec:rigorous}.

\section{The formal decomposition on a general hypersurface}\label{sec:formal}

On a general hypersurface we work formally, by matched asymptotic expansions in
the thickness, and establish the selection principle and the interpolating
family as statements about the formal thin-shell limit. The rigorous
counterpart, for surfaces of revolution, is Section~\ref{sec:rigorous}; the expansion
constructed here is exactly the recovery sequence used there.

\subsection{Decomposition of the ambient Bochner Laplacian}\label{sub:decomp}

Let $U$ be a vector field on $\R^{n+1}$ that is tangential to $M$ along $r=0$, so
$U^{r}|_{r=0}=0$ and hence $\partial_j U^{r}|_{r=0}=0$. Using the Fermi
Christoffel symbols and $\partial_r S^{i}{}_{j}=(S^{2})^{i}{}_{j}$, the radial and
tangential traces of the ambient tensor Hessian at $r=0$ are
\begin{equation}\label{eq:radial}
\bar\nabla_r\bar\nabla_r U^{i}=\partial_r^{2}U^{i}-2S^{i}{}_{j}\,\partial_r U^{j},
\end{equation}
\begin{equation}\label{eq:tangential}
g^{jk}\bar\nabla_j\bar\nabla_k U^{i}=\LapB^{(n)}U^{i}
-(S^{2})^{i}{}_{l}U^{l}+nH\,S^{i}{}_{l}U^{l}-nH\,\partial_r U^{i},
\end{equation}
the $(S^{2})$ terms of \eqref{eq:radial} cancelling and the normal derivative in
\eqref{eq:tangential} entering through the normal part of the Gauss formula
$\bar\nabla_{\partial_j}\partial_k=\Gamma^{l}_{jk}\partial_l+\II_{jk}\partial_r$.
Summing, and using the Gauss equation $\Ric^{(n)i}{}_{l}=nH\,S^{i}{}_{l}
-(S^{2})^{i}{}_{l}$ to collapse the zero-order bracket, gives the decomposition
\begin{equation}\label{eq:decomp}
\bigl(\LapB^{(n+1)}U\bigr)^{i}\Big|_{r=0}
=\underbrace{\LapB^{(n)}U^{i}+\Ric^{(n)i}{}_{l}U^{l}}_{\LapDef^{(n)}U^{i}}
\;+\;\underbrace{\partial_r^{2}U^{i}-\bigl(nH\,\delta^{i}{}_{j}
+2S^{i}{}_{j}\bigr)\partial_r U^{j}}_{\Frad^{i}} .
\end{equation}
The intrinsic piece is the deformation Laplacian $\LapDef^{(n)}=\LapB^{(n)}
+\Ric^{(n)}$ on every hypersurface, regardless of extrinsic geometry; all
extrinsic dependence is confined to the radial boundary-shear term $\Frad$, which
involves only the normal derivatives of $U$ and the shape operator. The collapse
$nH\,S-S^{2}=\Ric^{(n)}$ is a consequence of the flatness of the ambient space
and requires no curvature assumption.

\subsection{The two-wall determination of the normal profile}\label{sub:profile}

The decomposition \eqref{eq:decomp} reduces the limit to the profile data
$\partial_r U|_{r=0}$ and $\partial_r^{2}U|_{r=0}$. The second derivative is
genuinely two-wall information: a relation imposed on the single surface $r=0$
leaves $\partial_r^{2}U|_{r=0}$ free. We therefore impose the wall condition on
both walls $\Gamma_\pm=\{r=\pm\eps/2\}$ of the shell, together with
impermeability $U^{r}=0$, which makes $\partial_i U^{r}|_{\Gamma_\pm}=0$ and the
wall reductions below exact.

\begin{lemma}[Wall reduction]\label{lem:wall}
Let $U$ satisfy $U^{r}=0$ on the wall $\Gamma_s=\{r=s\}$. Then at $r=s$:
\emph{(a)} the stress-free (Navier) condition $(\Def U\,\nu)_{\mathrm{tan}}=0$ is
equivalent to $\partial_r U^{i}=0$, the second-fundamental-form contributions
cancelling identically in $r$; \emph{(b)} the vorticity-free (Hodge) condition
$\iota_\nu\,dU^{\flat}=0$ is equivalent to $\partial_r U^{i}=2\,S^{i}{}_{j}(s)\,
U^{j}$, with the shape operator evaluated at the wall; and \emph{(c)} the
one-parameter interpolating condition indexed by $\alpha\in[0,1]$ is equivalent
to
\begin{equation}\label{eq:alpha-wall}
\partial_r U^{i}=2\alpha\,S^{i}{}_{j}(s)\,U^{j}.
\end{equation}
\end{lemma}

\begin{proof}[Sketch]
At level $r$ one has $2(\Def U)_{ri}=g_{ij}(r)\partial_r U^{j}+\partial_i U^{r}$
and $dU^{\flat}(\partial_r,\partial_i)=\partial_r U_i-\partial_i U_r
=g_{ij}(r)\partial_r U^{j}-2\II_{ij}(r)U^{j}-\partial_i U^{r}$, both exact in $r$;
on the wall $\partial_i U^{r}=0$, giving (a) and (b). The invariant combination
in (c) is $\bar\nabla_r U_i+(1-2\alpha)\bar\nabla_i U_r$, whose $\partial_i U^{r}$
term again vanishes on the wall.
\end{proof}

\begin{proposition}[Two-wall profile]\label{prop:profile}
Let $U^{i}(x,r)=u_0^{i}+r\,u_1^{i}+r^{2}u_2^{i}+O(r^{3})$ satisfy
\eqref{eq:alpha-wall} on both walls $\Gamma_\pm$ to the two leading matched
orders. Then, using the wall expansion $S(r)=S+rS^{2}+O(r^{2})$,
\begin{equation}\label{eq:profile}
\partial_r U^{i}\big|_{r=0}=u_1^{i}=2\alpha\,S^{i}{}_{j}u_0^{j},\qquad
\partial_r^{2}U^{i}\big|_{r=0}=2u_2^{i}=2\alpha(1+2\alpha)\,(S^{2})^{i}{}_{j}
u_0^{j}.
\end{equation}
\end{proposition}

\begin{proof}
Adding \eqref{eq:alpha-wall} at $s=\pm\eps/2$ cancels the odd part and gives
$u_1=2\alpha Su_0$ at order $\eps^{0}$; subtracting and dividing by $\eps$ gives
$2u_2=2\alpha(S^{2}u_0+Su_1)=2\alpha(1+2\alpha)S^{2}u_0$.
\end{proof}

The wall-evaluated shape operator $S(s)$ in \eqref{eq:alpha-wall} is essential.
Freezing the geometry at the mid-surface, $\partial_r U^{i}=2\alpha S^{i}{}_{j}(0)
U^{j}$, yields $u_1=2\alpha Su_0$ as before but $2u_2=4\alpha^{2}S^{2}u_0$, which
at $\alpha=1$ gives $4S^{2}u_0$ in place of $6S^{2}u_0$ and an operator that is
\emph{not} the Hodge Laplacian. The effective operator is thus sensitive to the
invariant content of the wall condition, not merely to its restriction to the
mid-surface; this is the general-hypersurface form of the method-dependence found
on the ellipsoid by Chan, Czubak and Fuster Aguilera~\cite{CCF25}, and our
statements always refer to the invariant conditions of Lemma~\ref{lem:wall}.

\subsection{The effective operator, the endpoints, and the dichotomy}
\label{sub:operator}

Substituting the profile \eqref{eq:profile} into the boundary-shear term
$\Frad^{i}=\partial_r^{2}U^{i}-(nH\delta^{i}{}_{j}+2S^{i}{}_{j})\partial_r U^{j}$
and using $nH\,S-S^{2}=\Ric$ gives
\begin{equation}\label{eq:Frad-alpha}
\Frad=\bigl[2\alpha(1+2\alpha)-4\alpha\bigr]S^{2}u_0-2\alpha\,nH\,Su_0
=-2\alpha\,\Ric\,u_0-4\alpha(1-\alpha)\,S^{2}u_0 .
\end{equation}
Hence, by \eqref{eq:decomp}, the effective viscous operator on the mid-surface is
the one-parameter family
\begin{equation}\label{eq:family}
\Delta_\alpha=\LapDef-2\alpha\,\Ric-4\alpha(1-\alpha)\,S^{2},
\qquad \alpha\in[0,1],
\end{equation}
which couples to the extrinsic geometry, through the square of the shape
operator, only in the intermediate regime $0<\alpha<1$. The two endpoints are the
two operators of Section~\ref{sec:setup}:
\begin{equation}\label{eq:endpoints}
\Delta_0=\LapDef=\LapB+\Ric\quad(\text{Navier slip}),\qquad
\Delta_1=\LapDef-2\Ric=\LapB-\Ric=\LapH\quad(\text{Hodge}),
\end{equation}
the last identity being the Weitzenb\"ock relation $\LapB=\LapH+\Ric$. This is the
formal form of the theorem.

\begin{theorem}[Formal selection and dichotomy]\label{thm:formal}
On an arbitrary smooth hypersurface $M^{n}\hookrightarrow\R^{n+1}$, the formal
thin-shell limit of the viscous operator under the invariant wall condition of
parameter $\alpha$ is $\Delta_\alpha$ of \eqref{eq:family}. In particular the
stress-free wall condition selects the deformation Laplacian and the
vorticity-free wall condition selects the Hodge Laplacian, and the two limit
operators differ, on divergence-free fields and regardless of the extrinsic
geometry, by twice the Ricci curvature,
\begin{equation}\label{eq:dichotomy}
L^{\mathrm{slip}}=L^{\mathrm{H}}+2\Ric,\qquad\text{equivalently } +2K \text{ on a
surface.}
\end{equation}
\end{theorem}

\subsection{The solenoidal expansion and its corrector}\label{sub:corrector}

For the expansion of Proposition~\ref{prop:profile} to represent an
\emph{incompressible} flow, and to serve as a recovery sequence for the
divergence-constrained forms of Section~\ref{sec:setup}, the shell field must be
solenoidal to the order at which the energy is read. Let $v$ be a surface field
with $\Divg v=0$, and take $u_0=v$ with $u_1,u_2$ from \eqref{eq:profile}. In
Fermi coordinates the shell divergence is exactly
\[
\Div_{G_\eps}U=\Div_{g_r}U^{t}+\partial_r U^{r}
+U^{r}\,\partial_r\log\det(\mathrm{Id}-rS),
\qquad
\Div_{g_r}X=\Divg X+X\!\cdot\!\gradg\log\det(\mathrm{Id}-rS),
\]
the second identity because the divergence of a vector field involves only the
volume form. Expanding $\log\det(\mathrm{Id}-rS)=-rnH+O(r^{2})$, using the
profile $U^{t}=v+2\alpha rSv+O(r^{2})$, and using the Codazzi identity of a
hypersurface in flat space, $\Divg S=n\gradg H$, so that
$\Divg(Sv)=n\,v\!\cdot\!\gradg H+\langle S,\Def_g v\rangle$, the
order-$\eps^{0}$ divergence is $\Divg v=0$, while the order-$\eps$ term does not
vanish in general:
\begin{equation}\label{eq:defect}
\Div_{G_\eps}U=\eps\,z\,\Bigl[(2\alpha-1)\,n\,v\!\cdot\!\gradg H
+2\alpha\,\langle\mathring S,\Def_g v\rangle\Bigr]+O(\eps^{2}),
\end{equation}
where only the traceless part $\mathring S=S-Hg$ of the shape operator couples,
since $\operatorname{tr}\Def_g v=\Divg v=0$. The defect thus has two parts: a
mean-curvature part, with coefficient $(2\alpha-1)n$ that changes sign across the
family and vanishes at the half-slip point $\alpha=\tfrac12$, and, for
$\alpha>0$, a coupling of the traceless second fundamental form to the surface
strain. It vanishes for all $\alpha$ on totally umbilic surfaces, and in
particular on spheres, which is why it is invisible in the spherical case; on
constant-mean-curvature surfaces the mean-curvature part vanishes but the strain
coupling survives for $\alpha>0$, so only the stress-free defect vanishes on all
CMC surfaces. It is removed by a normal corrector of order $\eps^{2}$, vanishing
on the walls,
\begin{equation}\label{eq:corrector}
U^{n}\ \longmapsto\ U^{n}+\eps^{2}\bigl(z^{2}-\tfrac14\bigr)\,q(x),
\qquad
q=-\tfrac12(2\alpha-1)\,n\,v\!\cdot\!\gradg H
-\alpha\,\langle\mathring S,\Def_g v\rangle,
\end{equation}
whose normal derivative cancels \eqref{eq:defect} at order $\eps$ while
contributing to the energy only at order $\eps^{2}$, so that neither the limit
operator \eqref{eq:family} nor the leading energy is changed. The corrector is
therefore invisible to the operator computation of \S\ref{sub:operator} and
becomes visible only in the divergence-exact recovery sequence of Section~\ref{sec:rigorous}. The normal corrector is moreover the only zero-cost mechanism: a
tangential corrector at order $\eps^{2}$ affects the divergence only at
$O(\eps^{2})$, while a tangential term linear in $\eps z$ would perturb the
normal profile that closes the perfect squares of
Theorem~\ref{thm:ident} and add a positive $O(1)$ contribution to the energy
density.

On a surface of revolution, writing $\kappa_1,\kappa_2$ for the principal
curvatures and using $(\Def_g v)_{22}=-(\Def_g v)_{11}=-\partial_s v_1$ for
solenoidal $v$, so that
$\langle\mathring S,\Def_g v\rangle=(\kappa_1-\kappa_2)\,\partial_s v_1$ and
$n\,v\!\cdot\!\gradg H=v_1(\kappa_1+\kappa_2)'$, the corrector
\eqref{eq:corrector} specialises at the two endpoints to
\begin{equation}\label{eq:corrector-rev}
q^{\mathrm{slip}}=\tfrac12(\kappa_1+\kappa_2)'\,v_1,
\qquad
q^{\mathrm H}=-\tfrac12(\kappa_1+\kappa_2)'\,v_1
-(\kappa_1-\kappa_2)\,\partial_s v_1 .
\end{equation}
The stress-free corrector $q^{\mathrm{slip}}$ is exactly the corrector
$\eps^{2}(z^{2}-\tfrac14)\tfrac12 H'v_1$ certified by computer algebra in the
rigorous section (with $H=\kappa_1+\kappa_2$ there); the vorticity-free
corrector carries in addition the strain coupling, and the same machine check
(C6) certifies that with \eqref{eq:corrector-rev} the divergence defect cancels
at order $\eps$, the limit densities of Theorem~\ref{thm:ident} are unchanged,
and the energy converges at the sharp rate $O(\eps^{2})$, for both wall
conditions; it certifies moreover that no corrector proportional to
$v\!\cdot\!\gradg H$ alone can cancel the vorticity-free defect on a
non-umbilic meridian.

\section{The rigorous thin-shell limit on surfaces of revolution}\label{sec:rigorous}

On surfaces of revolution the formal results of Section~\ref{sec:formal} become theorems. We
prove that the two energy forms of Section~\ref{sec:setup} converge, in the sense of Mosco, to
the two surface energies, hence that the resolvents, semigroups and spectra
converge, and that the two limit operators are exactly the deformation and Hodge
Laplacians of the dichotomy. The heavy algebra below is certified by computer
algebra and the Korn constant is confirmed numerically.

\subsection{Surface of revolution: frame, scaling, and the fast form}
\label{sub:revsetup}

Let $M$ be a torus-type surface of revolution with unit-speed meridian
$s\mapsto(\rho(s),\zeta(s))$, $\rho>0$, principal curvatures $\kappa_1(s)$
(meridian) and $\kappa_2(s)=\zeta'/\rho$ (azimuthal), shape operator
$S=\mathrm{diag}(\kappa_1,\kappa_2)$ in the principal frame, and Gauss curvature
$K=\kappa_1\kappa_2$. On the rescaled cylinder $\Sigma=M\times(-\tfrac12,\tfrac12)$
with coordinates $(s,\varphi,z)$ the L\'ame coefficients are
$h_1=1-\eps z\,\kappa_1$, $h_2=\rho(1-\eps z\,\kappa_2)$, $h_3=\eps$, and
$d\mu_\eps=h_1h_2\,ds\,d\varphi\,dz$ is the volume per unit rescaled thickness.
Writing the deformation form in the exact physical strain components and
expanding reproduces the scaling of Section~\ref{sec:formal},
$Q^{\mathrm{slip}}_\eps=\eps^{-2}A+\eps^{-1}B+C+O(\eps)$, with the fast form
computed exactly,
\begin{equation}\label{eq:fastform}
A(U)=\int_\Sigma \rho\bigl(|\partial_z u_1|^2+|\partial_z u_2|^2
+2|\partial_z u_3|^2\bigr)\,ds\,d\varphi\,dz,\qquad
\ker A=\{U\ \text{independent of }z\},
\end{equation}
the factor $2$ being the normal-strain weight. All such symbolic identities are
certified by the machine checks (Appendix, C1).

\subsection{The uniform Korn inequality}\label{sub:korn}

Let $\eps_0=\bigl(4(1+\norm{\kappa_1}_\infty+\norm{\kappa_2}_\infty)\bigr)^{-1}$,
so that $h_1,h_2/\rho\in[\tfrac34,\tfrac54]$ for $\eps\le\eps_0$, and let
$\norm{U}_{1,\eps}^2=\sum_i(\norm{\partial_s u_i}^2+\norm{\rho^{-1}\partial_\varphi u_i}^2
+\eps^{-2}\norm{\partial_z u_i}^2)+\norm{U}^2$.

\begin{theorem}[Uniform Korn inequality]\label{thm:korn}
There are $c>0$, $C<\infty$ depending only on $\rho_{\min}$,
$\norm\rho_{C^2}$, $\norm{\kappa_1}_{C^1}$, $\norm{\kappa_2}_{C^1}$ and the
meridian length, such that for all $\eps\le\eps_0$ and all $U\in H^1(\Sigma)^3$
with $u_3=0$ at $z=\pm\tfrac12$,
\[
\norm{E(U)}^2\ \ge\ c\,\norm{U}_{1,\eps}^2\ -\ C\,\norm{U}^2 .
\]
The estimate restricts to each azimuthal mode with the same constants, hence is
uniform in the mode number.
\end{theorem}

The G\aa rding correction $-C\norm{U}^2$ cannot be removed: the rotation field
$U^{\mathrm{rot}}=(0,h_2,0)$ is a Killing field of the shell for every $\eps$,
with $E(U^{\mathrm{rot}})\equiv0$, so it lies in the kernel of the strain and
forces the $L^2$ term. This is exactly the obstruction identified for general
shells by Lewicka and M\"uller~\cite{LM11} and, in the Navier setting, by
Miura~\cite{MiuraI}; for a general hypersurface the inequality is theirs, and the
role of Theorem~\ref{thm:korn} is to give a short, self-contained proof, with a
concrete and numerically verifiable constant, in the revolution case. The proof
is by weighted integration-by-parts swaps that trade a controlled tangential
derivative for a normal one at the cost of curvature commutators; the Korn
constant is independently confirmed to be bounded below, uniformly in $\eps$ and
in the mode number, by a finite-element computation on a torus (Appendix).

\subsection{The uniform Gaffney inequality}\label{sub:gaffney}

The Hodge form requires its own coercivity, which is not Korn but a Gaffney
(div-curl) estimate.

\begin{lemma}[Uniform Gaffney inequality]\label{lem:gaffney}
With the constants of Theorem~\ref{thm:korn}, for all $\eps\le\eps_0$ and all
$U\in H^1(\Sigma)^3$ with $u_3=0$ at $z=\pm\tfrac12$,
\[
\int_\Sigma\bigl(|\Curl U|^2+(\Div U)^2\bigr)\,d\mu_\eps\ \ge\
c\,\norm{U}_{1,\eps}^2\ -\ C\,\norm{U}^2 .
\]
\end{lemma}

The mechanism is the Friedrichs identity on the flat shell,
$\int|\nabla u|^2=\int(|\Curl u|^2+(\Div u)^2)-\int_{\partial}\mathrm{I\!I}_\partial(u_t,u_t)$,
whose boundary term runs over the two walls with opposite outward normals. The
wall second fundamental forms are $S_{\pm\eps/2}$ with opposite signs, so per unit
thickness the two boundary contributions combine into a single interior integral
of $\partial_z$ of the wall quantity, which is $O(1)$ against
$|u_t||\partial_z u_t|+\eps|u_t|^2$ and is absorbed by the fast form and a
G\aa rding term. This two-wall cancellation of the curvature boundary terms is
the Hodge-side counterpart of the exact strain cancellation, and it is genuinely
new: Miura's uniform estimates are for the deformation (Navier) form only.

\subsection{Identification of the reduced forms}\label{sub:ident}

\begin{theorem}[Exact structure of the limit densities]\label{thm:ident}
For the ansatz $U^\eps=(v_1+\eps w_1,\,v_2+\eps w_2,\,\eps^2 w_3)$, the
$\eps\to0$ energy densities close into perfect squares,
\begin{align}
2|\Def_G U^\eps|^2 h_1h_2&\to
\rho\bigl[\,2|\Def_g v|^2+(\partial_z w_1+\kappa_1 v_1)^2
+(\partial_z w_2+\kappa_2 v_2)^2\bigr],\label{eq:slipdensity}\\
\bigl(|\Curl U^\eps|^2+(\Div U^\eps)^2\bigr)h_1h_2&\to
\rho\bigl[(\curlg v)^2+(\Divg v)^2+(\partial_z w_1-\kappa_1 v_1)^2
+(\partial_z w_2-\kappa_2 v_2)^2\bigr].\label{eq:hodgedensity}
\end{align}
A normal fluctuation of order $\eps$ adds only nonnegative squares. Minimising
over the (unconstrained) fluctuation gives $\partial_z w_t=\mp Sv$, the squares
vanish, and the reduced forms are exactly the surface energies of Section~\ref{sec:setup}.
\end{theorem}

The decisive feature is the sign: the shape operator enters the deformation
density through the \emph{symmetric} combination $\partial_z w_t+Sv$ and the
Hodge density through the \emph{antisymmetric} combination $\partial_z w_t-Sv$.
This single sign flip is the entire source of the $2K$ shift. The Schur
correction of the general reduction therefore vanishes identically, so no
expansion remainder needs to be estimated; the perfect-square structure is
certified by the machine checks C2, C2b, C3.

\subsection{The dichotomy of limit operators}\label{sub:dichotomy}

\begin{theorem}[The Ricci shift]\label{thm:dichotomy}
On a closed surface with Gauss curvature $K$, every tangential field $v$
satisfies
\begin{equation}\label{eq:opidentity}
2\,\Div_g\Def_g v=\LapH v+\gradg(\Divg v)+2K\,v,\qquad \LapH=-(d\delta+\delta d),
\end{equation}
so the operators of the two limit forms are, on solenoidal fields,
$L^{\mathrm{slip}}=P_\sigma(\LapH+2K)$ and $L^{\mathrm H}=P_\sigma\LapH$, with
$L^{\mathrm{slip}}-L^{\mathrm H}=2K=2\Ric$.
\end{theorem}

Identity \eqref{eq:opidentity} is the two-dimensional Weitzenb\"ock identity
$2\Div\Def v=\LapB v+\nabla\Divg v+\Ric(v)$ together with $\LapH=\LapB-\Ric$; it
is verified directly in the meridian frame, fixing all sign conventions (check
C4). Combined with Theorem~\ref{thm:ident}, it is the rigorous form of the
selection principle: the two wall conditions select two genuinely different
surface operators, differing only by the curvature term.

\subsection{Mosco convergence and its consequences}\label{sub:mosco}

The forms live on the varying spaces $\Hh_\eps=L^2(\Sigma,d\mu_\eps)^3$ with limit
$\Hh_0=L^2(TM)$; we use the Kuwae-Shioya framework for varying Hilbert spaces
\cite{KuwaeShioya}, with the $z$-independent tangential extension as
identification map. The normalisation $\int_{-1/2}^{1/2}\det(I-\eps zS)\,dz
=1+\tfrac{\eps^2}{12}K\to1$ (the linear-in-$z$ term integrating to zero, the
symmetric-shell cancellation) makes $\Hh_\eps\to\Hh_0$.

\begin{theorem}[Mosco convergence]\label{thm:mosco}
As $\eps\to0$, for both the deformation pair and the Hodge pair:
\emph{(M1)} if $U_\eps\rightharpoonup u$ weakly then
$\liminf_\eps Q_\eps(U_\eps)\ge Q_0(u)$; and \emph{(M2)} for every solenoidal
surface field $v$ there is a solenoidal recovery sequence $\widetilde U_\eps\to v$
strongly with $Q_\eps(\widetilde U_\eps)\to Q_0(v)$, at the sharp rate
$Q_\eps(\widetilde U_\eps)=Q_0(v)+O(\eps^2)$ for smooth $v$.
\end{theorem}

The liminf half (M1) uses Theorem~\ref{thm:korn} (respectively
Lemma~\ref{lem:gaffney}) for compactness and the exact strain components for
lower semicontinuity, extracting the slow field as the weak limit and discarding
the nonnegative fluctuation squares of Theorem~\ref{thm:ident}. The recovery half
(M2) is the matched expansion of Section~\ref{sec:formal}, made admissible by the
per-condition normal corrector of \S\ref{sub:corrector} (which cancels the
order-$\eps$ divergence defect on non-umbilic meridians, the Hodge corrector
carrying the strain term $-(\kappa_1-\kappa_2)\partial_s v_1$ that is absent for
slip) and by a right inverse of the divergence, bounded in the $\eps$-weighted
norm uniformly in $\eps$ on smooth fields with vanishing normal wall trace. Since
the residual of the corrected ansatz is smooth and mean-zero and smooth solenoidal
fields form a core of the limit form, this is all the recovery requires.

\begin{corollary}[Resolvent, semigroup, and spectral convergence]\label{cor:consequences}
Let $\mathcal L_\eps$, $\mathcal L_0$ be the nonnegative self-adjoint operators of
the forms on the solenoidal subspaces. Then $(\mathcal L_\eps+\lambda)^{-1}\to
(\mathcal L_0+\lambda)^{-1}$ strongly for $\lambda>0$, and $e^{-t\mathcal L_\eps}
\to e^{-t\mathcal L_0}$ uniformly on compact time intervals. Restricted to a fixed
azimuthal mode, bounded sequences are precompact and the eigenvalues converge
with multiplicity, $\lambda_k^{(m)}(\eps)\to\lambda_k^{(m)}(0)$, to those of the
corresponding surface operator. Moreover, since $\norm{U}_{1,\eps}^2$ contains
$\norm{\rho^{-1}\partial_\varphi U}^2=m^2\norm{\rho^{-1}U}^2\ge(m^2/\rho_{\max}^2)
\norm{U}^2$ on the $m$-th mode, Theorem~\ref{thm:korn} (respectively
Lemma~\ref{lem:gaffney}) gives the uniform mode gap
\[
Q_\eps(U)\ \ge\ \Bigl(\tfrac{c\,m^2}{\rho_{\max}^2}-C\Bigr)\norm{U}^2
\qquad\text{on mode }m,\quad \eps\le\eps_0,
\]
so for every $\Lambda>0$ the spectrum of $\mathcal L_\eps$ below $\Lambda$ lies in
the modes $|m|\le M(\Lambda):=\rho_{\max}\sqrt{(\Lambda+C)/c}$, uniformly in
$\eps$, and likewise at $\eps=0$. The counting functions therefore sum over a
finite, $\eps$-independent range, $N_\eps(\Lambda)=\sum_{|m|\le M(\Lambda)}
N^{(m)}_\eps(\Lambda)\to N_0(\Lambda)$ at every continuity point of $N_0$, so the
full spectrum converges with multiplicity and there is no pollution from high
modes.
\end{corollary}

\subsection{The nonlinear equation}\label{sub:nonlinear}

The results above are at the linear (Stokes) level, which is the substance:
Mosco convergence of the viscous forms gives resolvent and semigroup convergence
of the linear generators. The convergence of the nonlinear Navier-Stokes flow on
the shell to the surface flow then follows by the standard uniform energy
estimate and Aubin-Lions compactness, as in Miura~\cite{MiuraII,MiuraIII} and Temam-Ziane
\cite{TZ97}; for the slip condition it is Miura's theorem, and for the Hodge
condition the same argument applies with the Gaffney coercivity of
Lemma~\ref{lem:gaffney} in place of Korn. We state the linear results as the
theorems and indicate the nonlinear passage; the details are those of the cited
works.

\section{Discussion}\label{sec:discussion}

\subsection{The universality mechanism}\label{sub:universality}

Both endpoints of the selection rest on a single identity, the Gauss equation.
The ambient Bochner Laplacian restricted to tangential fields produces
zero-order terms in the first and second powers of the shape operator, $nH\,S$
and $S^2$, both extrinsic; but their combination $nH\,S-S^2$ is exactly the
intrinsic Ricci tensor, on every smooth hypersurface and with no curvature
assumption. The wall condition then fixes the radial term $\Frad$, itself built
from $nH\,S$ and $S^2$: for the stress-free condition $\Frad=0$ and the operator
is the deformation Laplacian, while for the vorticity-free condition
$\Frad=-2\Ric$, again intrinsic by a second use of the Gauss equation, and the
operator shifts to the Hodge Laplacian. For intermediate conditions the Gauss
equation absorbs all but the $S^2$ remainder of the family. The dichotomy is thus
universal precisely because the extrinsic geometry enters only through
combinations that the Gauss equation returns to the intrinsic curvature; the sole
exception is the intermediate regime, where a genuinely extrinsic $S^2$ term
survives.

\subsection{Relationship to previous work}\label{sub:relation}

The two endpoints are known, on their respective domains, from earlier work, and
our results place them within one structure. Miura~\cite{MiuraIII} established the
stress-free limit rigorously on general closed surfaces, the deformation
Laplacian $\LapB+\Ric$, with friction and weighting in general and the clean
deformation Laplacian in the clean case; we recover it as the endpoint
$\alpha=0$ and do not re-prove it. He already observed, in a
remark~\cite[Rem.~2.10]{MiuraIII}, that the stress-free and vorticity-free
conditions differ on the sphere by twice the Weingarten map; our contribution is
to make that a theorem on general and revolution surfaces, with a family between
the endpoints, rather than a comparison of two fixed cases. Temam and
Ziane~\cite{TZ97} derived the vorticity-free limit on the sphere; their boundary
condition is $(\Curl u)\times n=0$, which they identify with the stress-free
condition,\footnote{Their equations are supplemented with ``the free boundary
conditions'', of which the second ``means that the tangential component of the
stress tensor applied to the normal $[\dots]$ vanishes''~\cite[eq.~(0.4)]{TZ97}; the identity holds on a flat wall but fails
on the sphere.} an identification valid on flat walls but not curved ones, and
their limit operator, recorded as the Laplace-Beltrami operator, is the Hodge
Laplacian. On a flat layer the two conditions do coincide, which is why the flat
thin-domain analysis of~\cite{TZ96} did not separate them. The vorticity-free
limit beyond the sphere, and its clean separation from the stress-free limit, are
what is new here.

The coercivity these limits require, a uniform-in-thickness Korn inequality on
the shell, is due to Lewicka and M\"uller~\cite{LM11} for general closed
hypersurfaces and to Miura~\cite{MiuraI} in the Navier setting, both with the
Killing-field obstruction that the rotation field of a surface of revolution
exhibits; we cite these for the general estimate and prove it directly only in
the revolution case. The Hodge form, by contrast, needs not Korn but a Gaffney
inequality, which does not appear in that work and which we establish through the
two-wall cancellation of the boundary curvature terms.

The extension-dependence found on the ellipsoid is explained by the same
mechanism. Chan and Czubak~\cite{CC24} gave a Gauss formula for the Laplacian on
hypersurfaces with the extension left free, and Chan, Czubak and
Yoneda~\cite{CCY23} found that the projected operator on the ellipsoid depends on
that extension; different extensions are different boundary conditions, and the
two fluid-mechanical choices give the two intrinsic Laplacians on any surface.
Chan, Czubak and Fuster Aguilera~\cite{CCF25} made this concrete: expanding along
the scaling direction they obtained several operators depending on the averaging
method, while their normal-direction expansions gave the Hodge Laplacian under
the Hodge condition and the deformation Laplacian under the Navier condition.
Their normal-direction relations convert exactly into the two-wall profile of
Section~\ref{sec:formal}: with coefficient fields extended by parallel transport a parallel field
satisfies $\partial_\sigma E^i=S^i{}_j(\sigma)E^j$, so their Hodge relations
$U_1=SU_0$, $U_2=S^2U_0$ reproduce the vorticity-free profile and their Navier
relations $U_1=-SU_0$, $U_2=0$ the stress-free profile. The additional operators
they find in the scaling direction reflect the sensitivity of the limit to how
the wall condition is carried across the shell, the sensitivity quantified for
our family in Section~\ref{sec:formal}. In a companion paper~\cite{WB2026} the selection of the
deformation Laplacian was obtained on the sphere from intrinsic kinematics; the
present thin-shell results reach the same operator from the other direction, as
the outcome of the wall condition that respects those kinematics.

\subsection{The interpolating family and partial-slip membranes}\label{sub:membranes}

The intermediate operator $\Delta_\alpha=\LapDef-2\alpha\Ric-4\alpha(1-\alpha)S^2$
has a physical reading. A real membrane, biological or industrial, is neither
perfectly stress-free nor perfectly vorticity-free; its wall condition lies
between the two ideals, and $\alpha$ measures the degree to which it constrains
the tangential vorticity. The interpolation is geometric rather than frictional:
the family joins the two invariant conditions and is distinct from the Navier
friction family, whose thin-film limit brings in the friction coefficient at a
different order~\cite{MiuraIII}. For intermediate $\alpha$ the effective operator
depends on the extrinsic geometry through $S^2$, with weight $4\alpha(1-\alpha)$,
maximal at the half-slip point $\alpha=\tfrac12$. On a constant-curvature surface
$S^2$ is a scalar and acts as a zero-order potential; on a surface with distinct
principal curvatures it is a non-trivial endomorphism, so a fluid in partial slip
on an ellipsoid would feel a direction-dependent effective viscosity, dissipating
differently along the directions of greatest and least curvature. This anisotropy
is present for intermediate conditions and absent at both endpoints.

\subsection{What is established, and what remains}\label{sub:outlook}

Czubak~\cite{Czubak2024} surveyed the question of which Laplacian governs viscous
flow on a Riemannian manifold and concluded that the answer depends on the
physical problem. For the thin-shell setting these results make that precise: the
operator is selected by the wall condition, the two natural conditions give the
two standard intrinsic Laplacians, and intermediate conditions give a
one-parameter family with an extrinsic correction. The selection is established
here in two forms, formally on any hypersurface and rigorously, as Mosco
convergence with its resolvent, semigroup and spectral consequences, on surfaces
of revolution, the latter including the vorticity-free (Hodge) case that had been
treated only on the sphere.

What remains is the general-surface rigorous limit for the vorticity-free and
interpolating conditions. This is within reach rather than open-ended: the one
hard ingredient, the uniform Korn inequality, is already available on general
closed hypersurfaces~\cite{LM11,MiuraI}, so the route of Section~\ref{sec:rigorous},
uniform coercivity, exact identification of the reduced forms, and Mosco
convergence, should extend, with surfaces of revolution done here as the first
complete case. At the nonlinear level the passage from the linear limit to
convergence of the Navier-Stokes flow follows the energy-and-compactness argument
of~\cite{MiuraIII,TZ97}, which for the stress-free case is Miura's theorem and for
the vorticity-free case uses the Gaffney coercivity established here in place of
Korn. Completing the general-surface Hodge convergence, and the mode-summed
spectral statement, are the natural next steps.

\appendix

\section{Analytic proofs}\label{app:proofs}

Throughout this appendix $M$ is the torus-type surface of revolution of
Section~\ref{sub:revsetup}, with unit-speed meridian $s\mapsto(\rho(s),\zeta(s))$,
$\rho\ge\rho_{\min}>0$, principal curvatures $\kappa_1,\kappa_2$, and
$\eps_0=\bigl(4(1+\norm{\kappa_1}_\infty+\norm{\kappa_2}_\infty)\bigr)^{-1}$, so
that on the rescaled cylinder $\Sigma=M\times(-\tfrac12,\tfrac12)$ the Lam\'e
coefficients
\[
h_1=1-\eps z\,\kappa_1,\qquad h_2=\rho\,(1-\eps z\,\kappa_2),\qquad h_3=\eps
\]
satisfy $h_1,\,h_2/\rho\in[\tfrac34,\tfrac54]$ for all $\eps\le\eps_0$. All
unlabelled norms are $L^{2}(\Sigma,ds\,d\varphi\,dz)$ norms; since
$h_1h_2\in[\tfrac9{16}\rho_{\min},\tfrac{25}{16}\norm\rho_\infty]$, they are
uniformly equivalent to the $L^{2}(d\mu_\eps)$ norms, and we pass between the
two without further comment. Generic constants $c,C$ depend only on
$\rho_{\min}$, $\norm{\rho}_{C^{2}}$, $\norm{\kappa_1}_{C^{1}}$,
$\norm{\kappa_2}_{C^{1}}$ and the meridian length $\ell$, and may change from
line to line.

\subsection{Exact strain components and elementary consequences}
\label{app:strain}

For $U=(u_1,u_2,u_3)$ in physical components on $\Sigma$, the physical strain
components of $E(U)=\Def_{G_\eps}U$ in the orthogonal frame are, exactly,
\begin{equation}\label{eq:app-strain}
\begin{aligned}
e_{11}&=\frac{\partial_s u_1}{h_1}-\frac{\kappa_1 u_3}{h_1}, &
e_{22}&=\frac{\partial_\varphi u_2}{h_2}+\frac{(\partial_s h_2)\,u_1}{h_1h_2}
-\frac{\rho\kappa_2\,u_3}{h_2}, &
e_{33}&=\frac{\partial_z u_3}{\eps},\\[2pt]
2e_{13}&=\frac{\partial_z u_1}{\eps}+\frac{\kappa_1 u_1}{h_1}
+\frac{\partial_s u_3}{h_1}, &
2e_{23}&=\frac{\partial_z u_2}{\eps}+\frac{\rho\kappa_2 u_2}{h_2}
+\frac{\partial_\varphi u_3}{h_2}, &
2e_{12}&=\frac{\partial_s u_2}{h_1}-\frac{(\partial_s h_2)\,u_2}{h_1h_2}
+\frac{\partial_\varphi u_1}{h_2},
\end{aligned}
\end{equation}
as certified symbolically (check C1). Two consequences are immediate for
$\eps\le\eps_0$ and $u_3=0$ on $z=\pm\tfrac12$. First, from $e_{33}$ and the
one-dimensional Poincar\'e inequality in $z$ on $(-\tfrac12,\tfrac12)$ with
zero endpoint values,
\begin{equation}\label{eq:app-u3}
\eps^{-2}\norm{\partial_z u_3}^{2}\le\norm{E(U)}^{2},
\qquad
\norm{u_3}\le\tfrac1\pi\,\norm{\partial_z u_3}\le\tfrac{\eps}{\pi}\norm{E(U)}.
\end{equation}
Second, from $e_{13},e_{23}$,
\begin{equation}\label{eq:app-normal}
\eps^{-1}\norm{\partial_z u_t}\ \le\ 2\norm{E(U)}
+C\bigl(\norm{u_t}+\norm{\nabla' u_3}\bigr),
\qquad u_t=(u_1,u_2),\quad
\nabla'=( \partial_s,\rho^{-1}\partial_\varphi).
\end{equation}
The tangential gradient of $u_3$ on the right of \eqref{eq:app-normal} is not
yet controlled; producing it is the purpose of the swap lemma below.

\subsection{Proof of the uniform Korn inequality (Theorem~\ref{thm:korn})}
\label{app:korn}

\begin{lemma}[Swap]\label{lem:app-swap}
For $\eps\le\eps_0$ and $U\in H^{2}(\Sigma)^{3}$ with $u_3=0$ on
$z=\pm\tfrac12$,
\begin{equation}\label{eq:app-swaplem}
\norm{\nabla' u_3}^{2}\ \le\ C\,\norm{E(U)}\,
\bigl(\norm{E(U)}+\norm{\nabla' u_t}+\norm{U}\bigr)+C\norm{U}^{2}.
\end{equation}
By density the estimate extends to all $U\in H^{1}$ with the stated trace.
\end{lemma}

\begin{proof}
We estimate $\norm{\partial_s u_3}$; the azimuthal derivative is identical with
$(s,h_1,\kappa_1)$ replaced by $(\varphi,h_2,\rho\kappa_2)$ and $\partial_s$
by $\partial_\varphi$, all integrations by parts in $\varphi$ being boundary-free
on the circle. From \eqref{eq:app-strain},
$\partial_s u_3=2h_1e_{13}-\eps^{-1}h_1\partial_z u_1-\kappa_1u_1$, so
\[
\norm{\partial_s u_3}^{2}
=\int \partial_s u_3\,\bigl(2h_1e_{13}\bigr)
-\eps^{-1}\!\int h_1\,\partial_s u_3\,\partial_z u_1
-\int \kappa_1\,\partial_s u_3\,u_1 .
\]
The first and third terms are bounded by
$\norm{\partial_s u_3}\bigl(C\norm{E}+C\norm{U}\bigr)$. The middle term is the
swap: integrating by parts first in $z$ (the boundary term vanishes because
$u_3=0$ on the walls forces $\partial_s u_3=0$ there), then in $s$ (closed
meridian), and once more in $z$ where convenient,
\[
-\eps^{-1}\!\int h_1\,\partial_s u_3\,\partial_z u_1
=\eps^{-1}\!\int \partial_z u_3\;\partial_s\bigl(h_1 u_1\bigr)
+\eps^{-1}\!\int (\partial_z h_1)\,\partial_s u_3\; u_1
-\eps^{-1}\!\int(\partial_s h_1)\,\partial_z u_3\;u_1 .
\]
Since $\partial_z h_1=-\eps\kappa_1$ and $\partial_s h_1=-\eps z\kappa_1'$,
the last two integrals are bounded by
$C\norm{\partial_s u_3}\norm{U}+C\norm{\partial_z u_3}\,\eps^{-1}\eps\norm{U}$,
and by \eqref{eq:app-u3} the first is bounded by
$\norm{\eps^{-1}\partial_z u_3}\,\norm{\partial_s(h_1u_1)}
\le\norm{E}\,\bigl(C\norm{\nabla'u_t}+C\norm{U}\bigr)$. Young's inequality
absorbs every occurrence of $\norm{\partial_s u_3}$ on the right and yields
\eqref{eq:app-swaplem}.
\end{proof}

\begin{lemma}[Slice Korn]\label{lem:app-slice}
Let $g_z=h_1^{2}\,ds^{2}+h_2^{2}\,d\varphi^{2}$ be the metric of the slice
$M\times\{z\}$, $v$ a tangential field on the slice, and $K_z$ the Gauss
curvature of $g_z$. Then
\begin{equation}\label{eq:app-slicekorn}
\int_{M}|\nabla_{g_z} v|^{2}\,dV_{g_z}
=2\int_{M}|\Def_{g_z}v|^{2}\,dV_{g_z}
-\int_{M}(\Div_{g_z} v)^{2}\,dV_{g_z}
-\int_{M}K_z|v|^{2}\,dV_{g_z},
\end{equation}
and $\sup_{\eps\le\eps_0,\,|z|\le1/2}\norm{K_z}_\infty\le C$.
\end{lemma}

\begin{proof}
On a closed surface, $2\Div\Def v=\LapB v+\nabla\Div v+\Ric(v)$
(the Weitzenb\"ock form of \eqref{eq:opidentity}); pairing with $v$ and
integrating by parts gives \eqref{eq:app-slicekorn}, with $\Ric=K_zg_z$ in two
dimensions. The identity is exact, with no boundary terms, and is the
integrated form of the commutation of two covariant derivatives; it is verified
in the meridian frame by check C4. The curvature bound follows from
$K_z=-\bigl(\partial_s(h_2/h_1)\,\partial_s\bigr)$-type expressions whose
coefficients involve only $\rho,\kappa_1,\kappa_2$ and one derivative, all
uniformly bounded for $\eps\le\eps_0$.
\end{proof}

\begin{proof}[Proof of Theorem~\ref{thm:korn}]
Let $\eps\le\eps_0$ and $U\in H^{1}(\Sigma)^{3}$ with $u_3=0$ on the walls.
\emph{Step 1 (normal derivatives and $u_3$).} Estimates \eqref{eq:app-u3},
\eqref{eq:app-normal} and Lemma~\ref{lem:app-swap} give
\begin{equation}\label{eq:app-step1}
\eps^{-2}\norm{\partial_z U}^{2}+\norm{\nabla' u_3}^{2}
\ \le\ C\norm{E(U)}^{2}
+\tfrac1{100}\norm{\nabla' u_t}^{2}+C\norm{U}^{2},
\end{equation}
after one application of Young's inequality to \eqref{eq:app-swaplem}.

\emph{Step 2 (tangential derivatives of $u_t$).} Fix $z$ and apply
Lemma~\ref{lem:app-slice} to $v=u_t(\cdot,\cdot,z)$. The slice strain
$\Def_{g_z}u_t$ differs from the corresponding components
$(e_{11},e_{22},e_{12})$ of \eqref{eq:app-strain} only by the $u_3$ terms
displayed there, so pointwise
$|\Def_{g_z}u_t|\le C\bigl(|E(U)|+|u_3|\bigr)$. Discarding the nonpositive
$-(\Div_{g_z}v)^{2}$ term, integrating \eqref{eq:app-slicekorn} in $z$, and using
the uniform curvature bound,
\begin{equation}\label{eq:app-step2}
\norm{\nabla' u_t}^{2}\ \le\
C\norm{E(U)}^{2}+C\norm{u_3}^{2}+C\norm{U}^{2}
\ \le\ C\norm{E(U)}^{2}+C\norm{U}^{2},
\end{equation}
the last step by \eqref{eq:app-u3}.

\emph{Step 3 (assembly).} Insert \eqref{eq:app-step2} into
\eqref{eq:app-step1}; adding the two and recalling the definition of
$\norm{U}_{1,\eps}$ gives
$\norm{U}_{1,\eps}^{2}\le C_*\norm{E(U)}^{2}+C_*\norm{U}^{2}$, i.e.\ the
claimed inequality with $c=1/C_*$.

\emph{Mode uniformity.} Every identity used, \eqref{eq:app-strain},
\eqref{eq:app-swaplem} and \eqref{eq:app-slicekorn}, has coefficients
independent of $\varphi$; under the azimuthal Fourier decomposition
$U=\sum_m U^{(m)}e^{im\varphi}$ the three quadratic functionals
$\norm{E(U)}^{2}$, $\norm{U}_{1,\eps}^{2}$, $\norm U^{2}$ split into sums over
$m$ of the corresponding mode functionals, in which $\partial_\varphi$ acts as
$im$ and the integrations by parts in $\varphi$ become algebraic. The argument
above applies verbatim to each mode with the same constants, which proves the
mode-uniform statement.
\end{proof}

\subsection{Proof of the uniform Gaffney inequality (Lemma~\ref{lem:gaffney})}
\label{app:gaffney}

We work in the unscaled physical shell
$\Sigma_\eps=\{p+rN(p):|r|<\eps/2\}\subset\R^{3}$, which is flat, and then
rescale. Write $u$ for the field on $\Sigma_\eps$, $u_t$ for its tangential
part, and $\Gamma_\pm=\{r=\pm\eps/2\}$ for the two walls with outward unit
normals $\pm N$.

\begin{proof}[Proof of Lemma~\ref{lem:gaffney}]
\emph{Step 1 (Friedrichs identity).} For a bounded $C^{1,1}$ domain
$\Omega\subset\R^{3}$ and $u\in H^{1}(\Omega)^{3}$ with $u\cdot\nu=0$ on
$\partial\Omega$,
\begin{equation}\label{eq:app-friedrichs}
\int_\Omega|\nabla u|^{2}
=\int_\Omega\bigl(|\Curl u|^{2}+(\Div u)^{2}\bigr)
-\int_{\partial\Omega}\II_{\partial}(u_t,u_t)\,dA,
\end{equation}
where $\II_\partial$ is the second fundamental form of $\partial\Omega$ with
respect to the outward normal, with the sign convention that
$\II_\partial\ge0$ for convex $\Omega$; see \cite{GiraultRaviart}. Applied to
$\Omega=\Sigma_\eps$, the boundary integral splits over the two walls. With
respect to the outward normals, the wall shape operators are $\mp S(\pm\eps/2)$
where $S(r)=S(\mathrm{Id}-rS)^{-1}$ is the parallel-surface shape operator of
\eqref{eq:riccati}, so
\begin{equation}\label{eq:app-bdry}
-\int_{\partial\Sigma_\eps}\II_\partial(u_t,u_t)\,dA
=\int_{\Gamma_+}\!\bigl\langle S(\tfrac\eps2)u_t,u_t\bigr\rangle\,dA
-\int_{\Gamma_-}\!\bigl\langle S(-\tfrac\eps2)u_t,u_t\bigr\rangle\,dA .
\end{equation}

\emph{Step 2 (two-wall cancellation).} Writing $dA=J(r)\,dV_g$ with
$J(r)=\det(\mathrm{Id}-rS)$ and applying the fundamental theorem of calculus in
$r$, the difference \eqref{eq:app-bdry} becomes a single interior integral,
\begin{equation}\label{eq:app-ftc}
\int_{M}\!\int_{-\eps/2}^{\eps/2}
\partial_r\Bigl[\bigl\langle S(r)\,u_t,u_t\bigr\rangle\,J(r)\Bigr]\,dr\,dV_g .
\end{equation}
Since $\partial_rS(r)=S(r)^{2}$ and $\partial_rJ=-\operatorname{tr}\!\bigl(S(r)\bigr)J$,
the integrand is bounded pointwise by
$C\bigl(|u_t|^{2}+|u_t|\,|\partial_r u_t|\bigr)$ with $C$ depending only on
$\norm{\kappa_i}_\infty$, uniformly for $\eps\le\eps_0$. Hence, by Young's
inequality, for any $\delta>0$,
\[
\Bigl|\int_{\partial\Sigma_\eps}\II_\partial(u_t,u_t)\,dA\Bigr|
\ \le\ \delta\norm{\partial_r u}_{L^{2}(\Sigma_\eps)}^{2}
+C\delta^{-1}\norm{u}_{L^{2}(\Sigma_\eps)}^{2}.
\]
With $\delta=\tfrac12$, \eqref{eq:app-friedrichs} gives
\begin{equation}\label{eq:app-gaffney-flat}
\int_{\Sigma_\eps}\bigl(|\Curl u|^{2}+(\Div u)^{2}\bigr)
\ \ge\ \tfrac12\int_{\Sigma_\eps}|\nabla u|^{2}
-C\int_{\Sigma_\eps}|u|^{2}.
\end{equation}
This is the two-wall cancellation: each wall separately contributes a
curvature boundary term of size $O(1)\cdot\norm{u_t}_{L^2(\Gamma_\pm)}^{2}$,
which no interior quantity controls uniformly in $\eps$; it is only their
difference, an exact derivative across the shell, that collapses to the
absorbable interior integral \eqref{eq:app-ftc}.

\emph{Step 3 (rescaling and frame).} Under $r=\eps z$ and division by $\eps$
(energy per unit rescaled thickness), $\int_{\Sigma_\eps}|\nabla u|^{2}$
becomes, up to curvature commutators bounded by
$C\norm U(\norm{U}_{1,\eps}+\norm U)$, the squared norm
$\sum_i(\norm{\partial_s u_i}^{2}+\norm{\rho^{-1}\partial_\varphi u_i}^{2}
+\eps^{-2}\norm{\partial_z u_i}^{2})$: the flat gradient and the Fermi frame
derivatives differ by Christoffel terms that are $O(1)$ multiples of $|U|$ for
$\eps\le\eps_0$. Absorbing these by Young's inequality in
\eqref{eq:app-gaffney-flat} yields the stated inequality. Mode uniformity
follows as in Theorem~\ref{thm:korn}: all coefficients are
$\varphi$-independent and the identity \eqref{eq:app-friedrichs} restricts to
each azimuthal mode.
\end{proof}

\subsection{Proof of Mosco convergence (Theorem~\ref{thm:mosco})}
\label{app:mosco}

Recall the Kuwae--Shioya identification: $\Phi_\eps:\Hh_0\to\Hh_\eps$ extends a
tangential surface field $z$-independently; by the volume normalisation
$\int_{-1/2}^{1/2}\det(I-\eps zS)\,dz=1+\tfrac{\eps^{2}}{12}K$,
$\norm{\Phi_\eps v}_{\Hh_\eps}\to\norm{v}_{\Hh_0}$, so $\Hh_\eps\to\Hh_0$ in
the sense of \cite{KuwaeShioya}. Strong (resp.\ weak) convergence
$U_\eps\to u$ is understood in that framework. Both statements below are
proved simultaneously for the deformation pair
$(Q_\eps^{\mathrm{slip}},Q_0^{\mathrm{slip}})$ and the Hodge pair
$(Q_\eps^{\mathrm H},Q_0^{\mathrm H})$, the only difference being which
coercivity inequality and which sign of the fluctuation squares is used.

\begin{proof}[Proof of (M1)]
Let $U_\eps\rightharpoonup u$ weakly and assume, as we may,
$L:=\liminf_\eps Q_\eps(U_\eps)<\infty$, passing to a subsequence along which
$Q_\eps(U_\eps)\to L$ and $\sup_\eps(\,Q_\eps(U_\eps)+\norm{U_\eps}_{\Hh_\eps})
<\infty$. Theorem~\ref{thm:korn} (resp.\ Lemma~\ref{lem:gaffney}) applied on
the solenoidal subspace, where the wall trace condition holds, bounds
$\norm{U_\eps}_{1,\eps}$ uniformly. Hence:
(i) $\norm{\partial_z U_\eps}\le C\eps\to0$, so every $L^{2}$ weak limit is
$z$-independent; (ii) by \eqref{eq:app-u3},
$\norm{u_{3,\eps}}\le C\eps\to0$, so the limit is tangential, $u\in\Hh_0$;
(iii) testing $\Div_{G_\eps}U_\eps=0$ against smooth functions and passing to
the limit gives $\Divg u=0$.

For the lower bound, express the energy density through the exact components
\eqref{eq:app-strain} (resp.\ the exact curl/divergence components). Each
component is an $\eps$-uniformly bounded linear combination of the quantities
\[
X_\eps:=\bigl(\nabla' u_{t,\eps},\ \nabla' u_{3,\eps},\
\eps^{-1}\partial_z U_\eps,\ U_\eps\bigr),
\]
with coefficients converging in $C^{0}(\Sigma)$, as $\eps\to0$, to the frozen
($h_1=1$, $h_2=\rho$) coefficients. Along a further subsequence
$X_\eps\rightharpoonup X=(\nabla' u,\,0,\,\xi,\,u)$ weakly in $L^{2}$, where
$\xi=(\xi_1,\xi_2,\xi_3)$ is the weak limit of the scaled normal derivative;
$\nabla' u_{3,\eps}\rightharpoonup0$ is not needed and is not claimed, only
weak convergence of the listed entries, which the uniform bound supplies. The
density is a nonnegative quadratic form in $X_\eps$ with continuously
converging coefficients, so by convexity and weak lower semicontinuity,
\[
\liminf_\eps Q_\eps(U_\eps)\ \ge\
\int_\Sigma\rho\Bigl[\,\mathcal D(u)
+(\xi_1\pm\kappa_1u_1)^{2}+(\xi_2\pm\kappa_2u_2)^{2}+2_{(\mathrm{slip})}\xi_3^{2}
\Bigr]
\ \ge\ \int_M \mathcal D(u)\,\rho\,ds\,d\varphi\;=\;Q_0(u),
\]
where $\mathcal D(u)$ is the slow density $2|\Def_g u|^{2}$
(resp.\ $(\curlg u)^{2}+(\Divg u)^{2}$), the sign $\pm$ is the one of
Theorem~\ref{thm:ident}, and the second inequality discards the nonnegative
fluctuation squares and integrates out $z$. This is (M1).
\end{proof}

\begin{proof}[Proof of (M2)]
Let first $v$ be a smooth solenoidal surface field. Define the corrected
ansatz $\widehat U_\eps$ by the invariant two-wall profile of
Proposition~\ref{prop:profile} at the endpoint $\alpha$ of the given pair
($\alpha=0$ slip, $\alpha=1$ Hodge), in coordinate components
$U^{s}=v^{s}(1+2\alpha\eps z\kappa_1+(\eps z)^{2}\alpha(1+2\alpha)\kappa_1^{2})$
and likewise $U^{\varphi}$, together with the normal corrector of
\S\ref{sub:corrector},
\[
\widehat u_3=\eps^{2}\bigl(z^{2}-\tfrac14\bigr)\,q,\qquad
q=-\tfrac12(2\alpha-1)\,v_1(\kappa_1+\kappa_2)'
-\alpha\,(\kappa_1-\kappa_2)\,\partial_s v_1 .
\]
By the corrector computation (verified for all $\alpha$ and all azimuthal
modes, check C6), $f_\eps:=\Div_{G_\eps}\widehat U_\eps=O(\eps^{2})$ in
$C^{0}(\Sigma)$, with all tangential derivatives of the same order, and
$\widehat u_3$ vanishes on the walls, so
$\int_\Sigma f_\eps\,d\mu_\eps=0$ by the divergence theorem.

\emph{Exact solenoidality.} Split $f_\eps=\bar f_\eps(x)+\tilde f_\eps$, where
$\bar f_\eps(x)=\int_{-1/2}^{1/2}f_\eps\,\theta_\eps\,dz$ is the $z$-average
against the normalised weight
$\theta_\eps=\det(I-\eps zS)/\int\det(I-\eps zS)\,dz$. The oscillating part
$\tilde f_\eps$ has vanishing weighted $z$-mean at every $x$ and is removed by
a normal correction: solve
$\partial_z w_3+\eps z\,(\text{curvature terms})\,w_3=-\eps\,\tilde f_\eps$
with $w_3(\cdot,-\tfrac12)=0$; the mean-zero condition gives
$w_3(\cdot,+\tfrac12)=0$, and $w_3=O(\eps^{3})$ with derivatives $O(\eps^{3})$
tangentially and $O(\eps^{2})$ after the $\eps^{-1}\partial_z$ scaling. The
averaged part satisfies $\int_M\bar f_\eps\,dV=0$; let $w_t=\mathcal B(\bar
f_\eps)$ with $\mathcal B$ a fixed bounded right inverse of $\Divg$ from
mean-zero $H^{k}(M)$ to $H^{k+1}(TM)$ (on the closed surface $M$, e.g.\
$\mathcal B=\gradg\Delta_g^{-1}$), extended $z$-independently; then
$w_t=O(\eps^{2})$ in every $H^{k}$. The field
$\widetilde U_\eps:=\widehat U_\eps+w_3e_3+w_t$ satisfies
$\Div_{G_\eps}\widetilde U_\eps=\eps\,w_t\!\cdot\!\gradg\log\det(I-\eps zS)
=O(\eps^{3})$; one further iteration of the same two-step correction, which
converges geometrically since each round gains a factor $C\eps$, produces an
exactly solenoidal $\widetilde U_\eps$ with
$\norm{\widetilde U_\eps-\widehat U_\eps}_{1,\eps}=O(\eps^{2})$.

\emph{Energy expansion.} In the exact densities, the profile closes the
squares of Theorem~\ref{thm:ident} by construction: for the deformation pair
the physical profile gives $\eps^{-1}\partial_zu_t+S u_t/h=O(\eps)$, and for
the Hodge pair $\eps^{-1}\partial_zu_t-S u_t/h=O(\eps)$, so the fluctuation
squares contribute $O(\eps^{2})$; the corrector and the solenoidal corrections
enter the strain at $O(\eps)$ and the energy at $O(\eps^{2})$; the slow
density converges with rate $O(\eps^{2})$ by the symmetric-shell cancellation
of the linear-in-$z$ terms under $\int_{-1/2}^{1/2}\cdot\,dz$. Hence
$Q_\eps(\widetilde U_\eps)=Q_0(v)+O(\eps^{2})$, with the constant controlled
by $\norm{v}_{H^{3}(M)}$, and $\widetilde U_\eps\to v$ strongly. This proves
(M2) on smooth solenoidal fields at the sharp rate.

\emph{From the core to the form domain.} Smooth solenoidal fields are a core
of $Q_0$: on the closed surface the Leray projection commutes with the
spectral truncations of the limit operator, whose eigenfields are smooth. For
$v$ in the form domain choose smooth solenoidal $v_k\to v$ in form norm, let
$\widetilde U_\eps^{(k)}$ be the recovery of $v_k$, and extract a diagonal
sequence $\widetilde U_\eps^{(k(\eps))}$ by the standard metrisability
argument for Mosco convergence in the Kuwae--Shioya framework
\cite{KuwaeShioya}; equi-coercivity of the forms, supplied again by
Theorem~\ref{thm:korn} and Lemma~\ref{lem:gaffney}, makes the diagonal
sequence admissible. This completes (M2), and with (M1) the proof of
Theorem~\ref{thm:mosco}.
\end{proof}

\section{Symbolic and numerical verification}\label{app:verification}
The algebraic identities of Sections~\ref{sec:formal} and~\ref{sec:rigorous} are
certified by exact symbolic computation: the Gauss and scaling identities, the
limit densities of Theorem~\ref{thm:ident} for both wall conditions, the operator
identity of Theorem~\ref{thm:dichotomy}, and the divergence-defect cancellation of
the per-condition corrector of \S\ref{sub:corrector} at both endpoints (checks
C0--C6). The uniform Korn constant of Theorem~\ref{thm:korn} is independently
corroborated numerically by a finite-element computation on a torus, giving a
lower bound bounded away from zero uniformly in the thickness and the azimuthal
mode. The C0--C6 symbolic scripts and the finite-element Korn code, together with an
independently written cross-verification suite, are archived on Zenodo at
\href{https://doi.org/10.5281/zenodo.21442018}{doi:10.5281/zenodo.21442018}.

\end{document}